\documentclass{sigplanconf}

\usepackage{latexsym}
\usepackage{amssymb}
\usepackage{amsmath}
\usepackage{amsfonts}
\usepackage{array}
\usepackage{graphicx}
\usepackage[utf8]{inputenc}
\usepackage{tikz}

\renewcommand{\Re}{\mathrm{Re}}
\renewcommand{\Im}{\mathrm{Im}}

\newenvironment{proof}{\paragraph{\emph{Proof.}}}{\hfill$\square$}
\newtheorem{theorem}{Theorem}

\newtheorem{proposition}{Proposition}
\newtheorem{definition}{Definition}
\newtheorem{example}{Example}

\begin{document}

\setlength{\pdfpageheight}{\paperheight}
\setlength{\pdfpagewidth}{\paperwidth}

\CopyrightYear{2016}
%\setcopyright{rightsretained}
\publicationrights{author-pays}
\conferenceinfo{LICS '16}{July 05-08, 2016, New York, NY, USA}
\copyrightdata{978-1-4503-4391-6/16/07} %/isbn
\copyrightdoi{2933575.2934538} %/doi

\title{Solvability of Matrix-Exponential Equations}

\authorinfo{Jo\"{e}l Ouaknine \thanks{Supported by the EPSRC.}\thanks{Supported by the ERC grant AVS-ISS (648701).}}
           {University of Oxford}
           {joel@cs.ox.ac.uk}
\authorinfo{Amaury Pouly $^\dagger$}
           {University of Oxford}
           {amaury.pouly@cs.ox.ac.uk}
\authorinfo{Jo\~{a}o Sousa-Pinto $^\star$\thanks{Supported by the ERC grant ALGAME (321171).}}
           {University of Oxford}
           {jspinto@cs.ox.ac.uk}
\authorinfo{James Worrell $^\star$}
           {University of Oxford}
           {jbw@cs.ox.ac.uk}

\maketitle

\begin{abstract}
We consider a continuous analogue of \cite{MultiplicativeMatrixEquations}'s and \cite{ABC}'s problem of solving multiplicative matrix equations. Given $k+1$ square matrices $A_{1}, \ldots, A_{k}, C$, all of the same dimension, whose entries are real algebraic, we examine the problem of deciding whether there exist non-negative reals $t_{1}, \ldots, t_{k}$ such that
\begin{align*}
\prod \limits_{i=1}^{k} \exp(A_{i} t_{i}) = C .
\end{align*}
We show that this problem is undecidable in general, but decidable under the assumption that the matrices $A_{1}, \ldots, A_{k}$ commute. Our results have applications to reachability problems for linear hybrid automata.

Our decidability proof relies on a number of theorems from algebraic and transcendental number theory, most notably those of Baker, Kronecker, Lindemann, and Masser, as well as some useful geometric and linear-algebraic results, including the Minkowski-Weyl theorem and a new (to the best of our knowledge) result about the uniqueness of strictly upper triangular matrix logarithms of upper unitriangular matrices. On the other hand, our undecidability result is shown by reduction from Hilbert's Tenth Problem.
\end{abstract}


%\terms{...}

\category{Theory of Computation}{Models of Computation}{Timed and hybrid models}
\category{Mathematics of computing}{Mathematical Analysis}{Ordinary differential equations}
\category{Mathematics of computing}{Mathematical Analysis}{Number-theoretic computations}
\section{Introduction}

Reachability problems are a fundamental staple of theoretical computer
science and verification, one of the best-known examples being the
Halting Problem for Turing machines. In this paper, our motivation
originates from systems that evolve continuously subject to linear
differential equations; such objects arise in the analysis of a range
of models, including linear hybrid automata, continuous-time Markov
chains, linear dynamical systems and cyber-physical systems as they
are used in the physical sciences and engineering---see,
e.g.,~\cite{Alu15}.

More precisely, consider a system consisting of a finite number of
discrete locations (or control states), having the property that the
continuous variables of interest evolve in each location according to
some linear differential equation of the form $\dot{\boldsymbol{x}} = A
\boldsymbol{x}$; here $\boldsymbol{x}$ is a vector of continuous
variables, and $A$ is a square `rate' matrix of appropriate
dimension. As is well-known, in each location the closed form solution
$\boldsymbol{x}(t)$ to the differential equation admits a
matrix-exponential representation of the form $\boldsymbol{x}(t) =
\exp(At)\boldsymbol{x}(0)$. Thus if a system evolves through a series
of $k$ locations, each with rate matrix $A_i$, and spending time $t_i
\geq 0$ in each location, the overall effect on the initial continuous
configuration is given by the matrix
\begin{align*}
\prod \limits_{i=1}^{k} \exp(A_{i} t_{i}) \, ,
\end{align*}
viewed as a linear transformation on $\boldsymbol{x}(0)$.\footnote{In
  this motivating example, we are assuming that there are no discrete
  resets of the continuous variables when transitioning between
  locations.}

A particularly interesting situation arises when the matrices $A_i$
commute; in such cases, one can show that the order in which the
locations are visited (or indeed whether they are visited only once or
several times) is immaterial, the only relevant data being the total
time spent in each location. Natural questions then arise as to what
kinds of linear transformations can thus be achieved by such systems.

\subsection{Related Work}

Consider the following problems, which can be seen as discrete analogues of the question we deal with in this paper.

\begin{definition}[Matrix Semigroup Membership Problem]
Given $k+1$ square matrices $A_{1}, \ldots, A_{k}, C$, all of the same dimension, whose entries are algebraic, does the matrix $C$ belong to the multiplicative semigroup generated by $A_{1}, \ldots, A_{k}$?
\end{definition}

\begin{definition}[Solvability of Multiplicative Matrix Equations]
Given $k+1$ square matrices $A_{1}, \ldots, A_{k}, C$, all of the same dimension, whose entries are algebraic, does the equation
\begin{align*}
\prod\limits_{i=1}^{k} A_{i}^{n_{i}} = C
\end{align*}
admit any solution $n_{1}, \ldots, n_{k} \in \mathbb{N}$?
\end{definition}

In general, both problems have been shown to be undecidable, in
\cite{Paterson} and \cite{MEHTP}, by reductions from Post's
Correspondence Problem and Hilbert's Tenth Problem, respectively.

When the matrices $A_{1}, \ldots, A_{k}$ commute, these problems are
identical, and known to be decidable, as shown in
\cite{MultiplicativeMatrixEquations}, generalising the solution of the
matrix powering problem, shown to be decidable in \cite{KL}, and the
case with two commuting matrices, shown to be decidable in \cite{ABC}.

See \cite{HalavaSurvey} for a relevant survey, and \cite{CK05} for
some interesting related problems.

The following continuous analogue of \cite{KL}'s Orbit Problem was
shown to be decidable in \cite{Hainry}:

\begin{definition}[Continuous Orbit Problem]
Given an $n \times n$ matrix $A$ with algebraic entries and two
$n$-dimensional vectors $\boldsymbol{x}, \boldsymbol{y}$ with
algebraic coordinates, does there exist a non-negative real $t$ such
that $\exp(At) \boldsymbol{x} = \boldsymbol{y}$?
\end{definition}

The paper \cite{ContinuousOrbitIPL} simplifies the argument of
\cite{Hainry} and shows polynomial-time decidability. Moreover, a
continuous version of the Skolem-Pisot problem was dealt with in
\cite{ContinuousSkolem}, where a decidability result is presented for
some instances of the problem.

As mentioned earlier, an important motivation for our work comes from
the analysis of hybrid automata. In addition to~\cite{Alu15},
excellent background references on the topic are
\cite{HenzingerSTOC,HenzingerLICS}.

\subsection{Decision Problems}

We start by defining three decision problems that will be the main
object of study in this paper: the \emph{Matrix-Exponential Problem},
the \emph{Linear-Exponential Problem}, and the
\emph{Algebraic-Logarithmic Integer Programming} problem.

\begin{definition}
  An instance of the Matrix-Exponential Problem (MEP) consists of
  square matrices $A_{1}, \ldots, A_{k}$ and $C$, all of the same
  dimension, whose entries are real algebraic numbers.  The problem
asks to determine whether there exist real numbers 
$t_1,\ldots,t_k \geq 0$ such that 
\begin{align}
\label{MEP}
\prod \limits_{i=1}^{k} \exp(A_{i} t_{i}) = C \, .
\end{align}
\label{def:MEP}
\end{definition}

We will also consider a generalised version of this problem, called
the \emph{Generalised MEP}, in which the matrices $A_1,\ldots,A_k$ and
$C$ are allowed to have complex algebraic entries and in which the
input to the problem also mentions a polyhedron
$\mathcal{P}\subseteq\mathbb{R}^{2k}$ that is specified by linear
inequalities with real algebraic coefficients.  In the generalised problem
we seek $t_1,\ldots,t_k \in \mathbb{C}$ that satisfy (\ref{MEP}) and
such that the vector
$(\Re(t_1),\ldots,\Re(t_k),
\Im(t_1),\ldots,\Im(t_k))$ lies in $\mathcal{P}$.

In the case of commuting matrices, the Generalised Matrix-Exponential
Problem can be analysed block-wise, which leads us to the following
problem:

\begin{definition}
  An instance of the Linear-Exponential Problem (LEP) consists of a system
  of equations
\begin{align}
\label{single_eqn_form}
  \exp\left(\sum_{i \in I} \lambda_i^{(j)} t_i \right) = c_j \exp (d_j) 
\quad (j \in J),
\end{align}
where $I$ and $J$ are finite index sets, the $\lambda_i^{(j)}$, $c_j$
and $d_j$ are complex algebraic constants, and the $t_i$ are complex
variables, together with a polyhedron
$\mathcal{P} \subseteq \mathbb{R}^{2k}$ that is specified by a system
of linear inequalities with algebraic coefficients.  The problem asks
to determine whether there exist $t_1,\ldots,t_k\in \mathbb{C}$ that
satisfy the system (\ref{single_eqn_form}) and such that
$(\Re(t_1),\ldots,\Re(t_k),\Im(t_1),\ldots,\Im(t_k))$
lies in $\mathcal{P}$.
\label{def:LEP}
\end{definition}

To establish decidability of the Linear-Exponential Problem, we reduce
it to the following 
\emph{Algebraic-Logarithmic Integer Programming}
problem.  Here a \emph{linear form in logarithms of algebraic numbers}
is a number of the form
$\beta_{0} + \beta_{1} \log(\alpha_{1}) + \cdots + \beta_{m}
\log(\alpha_{m})$,
where
$\beta_{0}, \alpha_{1}, \beta_{1}, \ldots, \alpha_{m}, \beta_{m}$ are
algebraic numbers and $\log$ denotes a fixed branch of the complex
logarithm function.

% \exp(\boldsymbol{\lambda} \cdot \boldsymbol{t}) = c_{\boldsymbol{\lambda}} \exp(d_{\boldsymbol{\lambda}}), \quad \boldsymbol{\lambda} \in \overline{\mathbb{Q}}^{n}, c_{\boldsymbol{\lambda}}, d_{\boldsymbol{\lambda}} \in \overline{\mathbb{Q}}
% \end{align}
% and a convex polyhedron $\mathcal{P} \subseteq \mathbb{R}^{2k}$ with an algebraic description, decide whether (\ref{single_eqn_form}) admits a solution $\boldsymbol{t}$ such that $(\Re(\boldsymbol{t}), \Im(\boldsymbol{t})) \in \mathcal{P}$.
% \end{definition}

% In both cases, when a constraint of the form $(\Re(\boldsymbol{t}), \Im(\boldsymbol{t})) \in \mathcal{P}$ is imposed, where $\mathcal{P} \subseteq \mathbb{R}^{2k}$ is a convex polyhedron with an algebraic description, we refer to the problem as the \emph{constrained MEP/LEP}; the predicates defined above may often be combined, e.g. \emph{constrained commuting MEP}.

\begin{definition}
An instance of the Algebraic-Logarithmic Integer Programming Problem (ALIP) consists of a finite system of equations of the form
\begin{align*}
A \boldsymbol{x} \leq \frac{1}{\pi} \boldsymbol{b}
\end{align*}
where $A$ is an $m\times n$ matrix with real algebraic entries and
where the coordinates of $\boldsymbol{b}$ are real linear forms in
logarithms of algebraic numbers. The problem asks to determine whether
such a system admits a solution $\boldsymbol{x} \in \mathbb{Z}^{n}$.
\end{definition}

\subsection{Paper Outline}

After introducing the main mathematical techniques that are used in
the paper, we present a reduction from the Generalised Matrix
Exponential Problem with commuting matrices to the Linear-Exponential
Problem, as well as a reduction from the Linear-Exponential Problem to
the Algebraic-Logarithmic Integer Programming Problem, before finally showing that the Algebraic-Logarithmic Integer Programming Problem is decidable. By way of hardness, we will prove that the Matrix-Exponential Problem is
undecidable (in the non-commutative case), by reduction from Hilbert's
Tenth Problem.

\section{Mathematical Background}
\label{background}

\subsection{Number Theory and Diophantine Approximation}

A number $\alpha \in \mathbb{C}$ is said to be \emph{algebraic} if
there exists a non-zero polynomial $p \in \mathbb{Q}[x]$ for which
$p(\alpha) = 0$. A complex number that is not algebraic is said to be
\emph{transcendental}. The monic polynomial $p \in \mathbb{Q}[x]$ of
smallest degree for which $p(\alpha) = 0$ is said to be the minimal
polynomial of $\alpha$. The set of algebraic numbers, denoted by
$\overline{\mathbb{Q}}$, forms a field. Note that the complex
conjugate of an algebraic number is also algebraic, with the same
minimal polynomial. It is possible to represent and manipulate
algebraic numbers effectively, by storing their minimal polynomial and
a sufficiently precise numerical approximation. An excellent course
(and reference) in computational algebraic number theory can be found
in \cite{Cohen}. Efficient algorithms for approximating algebraic
numbers were presented in \cite{Pan}.

Given a vector $\boldsymbol{\lambda} \in \overline{\mathbb{Q}}^{m}$, its \emph{group of multiplicative relations} is defined as
\begin{align*}
L(\boldsymbol{\lambda}) = \lbrace \boldsymbol{v} \in \mathbb{Z}^{m} : \boldsymbol{\lambda}^{\boldsymbol{v}} = 1 \rbrace .
\end{align*}

Moreover, letting $\log$ represent a fixed branch of the complex logarithm function, note that $\log(\alpha_{1}), \ldots, \log(\alpha_{m})$ are linearly independent over $\mathbb{Q}$ if and only if
\begin{align*}
L(\alpha_{1}, \ldots, \alpha_{m}) = \lbrace \boldsymbol{0} \rbrace .
\end{align*}

Being a subgroup of the free finitely generated abelian group $\mathbb{Z}^{m}$, the group $L(\boldsymbol{\lambda})$ is also free and admits a finite basis.

The following theorem, due to David Masser, allows us to effectively determine $L(\boldsymbol{\lambda})$, and in particular decide whether it is equal to $\lbrace \boldsymbol{0} \rbrace$. This result can be found in \cite{Masser}.

\begin{theorem}[Masser]
The free abelian group $L(\boldsymbol{\lambda})$ has a basis $\boldsymbol{v}_{1}, \ldots, \boldsymbol{v}_{l} \in \mathbb{Z}^{m}$ for which
\begin{align*}
\max\limits_{1 \leq i \leq l, 1 \leq j \leq m} \lvert v_{i,j} \rvert \leq (D \log H)^{O(m^{2})}
\end{align*}
where $H$ and $D$ bound respectively the heights and degrees of all the $\lambda_{i}$.
\end{theorem}

%We will need the following results of Baker~\cite{Baker75}.
%The first one, together with Masser's theorem, allows us
%to eliminate all algebraic relations in the description of linear
%forms in logarithms of algebraic numbers.

Together with the following result, due to Alan Baker, Masser's theorem allows us to eliminate all algebraic relations in the description of linear forms in logarithms of algebraic numbers. In particular, it also yields a method for comparing linear forms in logarithms of algebraic numbers: test whether their difference is zero and, if not, approximate it numerically to sufficient precision, so as to infer its sign. Note that the set of linear forms in logarithms of algebraic numbers is closed under addition and under multiplication by algebraic numbers, as well as under complex conjugation. See \cite{Baker75} and \cite{BakerPaper}.

\begin{theorem}[Baker]
Let $\alpha_{1}, \ldots, \alpha_{m} \in \overline{\mathbb{Q}} \setminus \lbrace 0 \rbrace$. If
\begin{align*}
\log(\alpha_{1}), \ldots, \log(\alpha_{m})
\end{align*}
are linearly independent over $\mathbb{Q}$, then
\begin{align*}
1, \log(\alpha_{1}), \ldots, \log(\alpha_{m})
\end{align*}
are linearly independent over $\overline{\mathbb{Q}}$.
\end{theorem}

%The next result essentially implies that one can effectively check
%whether a linear form in logarithms of algebraic numbers equals
%zero. Noting that the set of linear forms in logarithms of algebraic
%numbers is closed under addition and multiplication by algebraic
%numbers, it easily follows that one can effectively compare two linear
%forms in logarithms of algebraic numbers. It is also closed under
%complex conjugation. See \cite{Baker75} and \cite{BakerPaper}.

%\begin{theorem}[Baker]
%Let $\alpha_{1}, \ldots, \alpha_{m}$ be non-zero algebraic numbers with degrees at most $d$ and heights at most $A$. Further, let $\beta_{0}, \ldots, \beta_{m}$ be algebraic numbers with degrees at most $d$ and heights at most $B$, where $B \geq 2$. Write
%\begin{align*}
%\Lambda = \beta_{0} + \beta_{1} \log(\alpha_{1}) + \cdots + \beta_{m} \log(\alpha_{m}) .
%\end{align*}
%Then either $\Lambda = 0$ or $\lvert \Lambda \rvert > B^{-C}$, where $C$ is an effectively computable number depending only on $m$, $d$, $A$, and the chosen branch of the complex logarithm.
%\end{theorem}

The theorem below was proved by Ferdinand von Lindemann in 1882, and later generalised by Karl Weierstrass in what is now known as the Lindemann-Weierstrass theorem. As a historical note, this result was behind the first proof of transcendence of $\pi$, which immediately follows from it.

\begin{theorem}[Lindemann]
If $\alpha \in \overline{\mathbb{Q}} \setminus \lbrace 0 \rbrace$, then $e^{\alpha}$ is transcendental.
\end{theorem}

We will also need the following result, due to Leopold Kronecker, on simultaneous Diophantine approximation, which generalises Dirichlet's Approximation Theorem. We denote the \emph{group of additive relations} of $\boldsymbol{v}$ by
\begin{align*}
A(\boldsymbol{v}) = \lbrace \boldsymbol{z} \in \mathbb{Z}^{d} : \boldsymbol{z} \cdot \boldsymbol{v} \in \mathbb{Z} \rbrace .
\end{align*}

Throughout this paper, $\operatorname{dist}$ refers to the $l_{1}$ distance.

\begin{theorem}[Kronecker]
\label{Kronecker}
Let $\boldsymbol{\alpha}_{1}, \ldots, \boldsymbol{\alpha_{k}} \in \mathbb{R}^{d}$ and $\boldsymbol{\beta} \in \mathbb{R}^{d}$. The following are equivalent:
\begin{enumerate}
\item For any $\varepsilon > 0$, there exists $\boldsymbol{n} \in \mathbb{N}^{k}$ such that
\begin{align*}
\operatorname{dist}(\boldsymbol{\beta} + \sum\limits_{i=1}^{k} n_{i} \boldsymbol{\alpha}_{i}, \mathbb{Z}^{d}) \leq \varepsilon .
\end{align*}
\item It holds that
\begin{align*}
\bigcap\limits_{i=1}^{k} A(\boldsymbol{\alpha}_{i}) \subseteq A(\boldsymbol{\beta}) .
\end{align*}
\end{enumerate}
\end{theorem}

Many of these results, or slight variations thereof, can be found in \cite{HardyAndWright} and \cite{Cassels}.

\subsection{Lattices}

Consider a non-zero matrix $K\in\overline{\mathbb{Q}}^{r\times d}$ and vector
$\boldsymbol{k} \in \overline{\mathbb{Q}}^r$.  The following proposition shows
how to compute a representation of the affine lattice
$\{ \boldsymbol{x}\in\mathbb{Z}^d : K\boldsymbol{x} = \boldsymbol{k}
\}$.
Further information about lattices can be found in \cite{LatticeBook}
and \cite{Cohen}.

\begin{proposition}
There exist $\boldsymbol{x}_{0} \in \mathbb{Z}^{d}$ and 
$M \in \mathbb{Z}^{d \times s}$, where $s < r$, such that
\begin{align*}
  \{ \boldsymbol{x}\in\mathbb{Z}^d : K\boldsymbol{x} =
  \boldsymbol{k} \} = 
  \boldsymbol{x}_{0} + \{ M \boldsymbol{y} : \boldsymbol{y} \in \mathbb{Z}^s \} \, .
\end{align*}
\end{proposition}

\begin{proof}
  Let $\theta$ denote a primitive element of the number field
  generated by the entries of $K$ and $\boldsymbol{k}$. Let the degree
  of this extension, which equals the degree of $\theta$, be
  $D$. Then for $\boldsymbol{x} \in \mathbb{Z}^d$ one can write
\begin{align*}
K \boldsymbol{x} = \boldsymbol{k} &\Leftrightarrow \left( \sum \limits_{i=0}^{D-1} N_{i} \theta^{i} \right) \boldsymbol{x} = \sum \limits_{i=0}^{D-1} \boldsymbol{k}_{i} \theta^{i} \\
&\Leftrightarrow N_{i} \boldsymbol{x} = \boldsymbol{k}_{i}, \forall i \in \lbrace 0, \ldots, D-1 \rbrace ,
\end{align*}
for some integer matrices
$N_{0}, \ldots, N_{D-1} \in \mathbb{Z}^{r \times d}$ and integer
vectors
$\boldsymbol{k}_{0}, \ldots, \boldsymbol{k}_{D-1} \in \mathbb{Z}^{r}$.
The solution of each of these equations is clearly an affine lattice, and
therefore so is their intersection.
\end{proof}

\subsection{Matrix exponentials}

Given a matrix $A \in \mathbb{C}^{n \times n}$, its exponential is defined as
\begin{align*}
\exp(A) = \sum \limits_{i=0}^{\infty} \frac{A^{i}}{i!} .
\end{align*}
The series above always converges, and so the exponential of a matrix is always well defined. The standard way of computing $\exp(A)$ is by finding $P \in \mathit{GL}_{n}(\mathbb{C})$ such that $J=P^{-1}AP$ is in Jordan Canonical Form, and by using the fact that $\exp(A) = P \exp(J) P^{-1}$, where $\exp(J)$ is easy to compute. When $A \in \overline{\mathbb{Q}}^{n \times n}$, $P$ can be taken to be in $GL_{n}(\overline{\mathbb{Q}})$; note that

\begin{align*}
\mbox{if } J &= \begin{pmatrix}
\lambda && 1 && 0 && \cdots && 0 \\
0 && \lambda && 1 &&\cdots && 0 \\
\vdots && \vdots && \ddots && \ddots && \vdots \\
0 && 0 && \cdots && \lambda && 1 \\
0 && 0 && \cdots && 0 && \lambda
\end{pmatrix} \mbox{ then } \\
\exp(Jt) &= \exp(\lambda t) \begin{pmatrix}
1 && t && \frac{t^{2}}{2} && \cdots && \frac{t^{k-1}}{(k-1)!} \\
0 && 1 && t && \cdots && \frac{t^{k-2}}{(k-2)!} \\
\vdots && \vdots &&\ddots && \ddots && \vdots \\
0 && 0 && \cdots && 1 && t \\
0 && 0 && \cdots && 0 && 1
\end{pmatrix} .
\end{align*}

Then $\exp(J)$ can be obtained by setting $t=1$, in particular $\exp(J)_{ij} = \frac{\exp(\lambda)}{(j-i)!}$ if $j \geq i$ and $0$ otherwise.

When $A$ and $B$ commute, so must $\exp(A)$ and $\exp(B)$. Moreover, when $A$ and $B$ have algebraic entries, the converse also holds, as shown in \cite{MatrixExps}. Also, when $A$ and $B$ commute, it holds that $\exp(A)\exp(B) = \exp(A+B)$.

\subsection{Matrix logarithms}

The matrix $B$ is said to be a logarithm of the matrix $A$ if $\exp(B) = A$. It is well known that a logarithm of a matrix $A$ exists if and only if $A$ is invertible. However, matrix logarithms need not be unique. In fact, there exist matrices admitting uncountably many logarithms. See, for example, \cite{MatrixLogs1} and \cite{MatrixLogs2}.

A matrix is said to be unitriangular if it is triangular and all its diagonal entries equal $1$. Crucially, the following uniqueness result holds:

\begin{theorem}
\label{logarithm_uniqueness}
Given an upper unitriangular matrix $M \in \mathbb{C}^{n \times n}$, there exists a unique strictly upper triangular matrix $L$ such that $\exp(L)=M$. Moreover, the entries of $L$ lie in the number field $\mathbb{Q}(M_{i,j}: 1 \leq i,j \leq n)$.
\end{theorem}

\begin{proof}
Firstly, we show that, for any strictly upper triangular matrix $T$ and for any $1<m<n$ and $i<j$, the term $(T^{m})_{i,j}$ is polynomial on the elements of the set $\lbrace T_{r,s} : s-r<j-i \rbrace$. This can be seen by induction on $m$, as each $T^{m}$ is strictly upper triangular, and so
\begin{align*}
(T^{m})_{i,j} = \sum\limits_{l=1}^{n} (T^{m-1})_{i,l} T_{l,j} = \sum\limits_{l=i+1}^{j-1} (T^{m-1})_{i,l} T_{l,j} .
\end{align*}

Finally, we show, by induction on $j-i$, that each $L_{i,j}$ is polynomial on the elements of the set
\begin{align*}
\lbrace M_{i,j} \rbrace \cup \lbrace M_{r,s} : s-r < j-i \rbrace . 
\end{align*}
If $j-i \leq 0$, then $L_{i,j}=0$, so the claim holds. When $j-i>0$, as $L$ is nilpotent,
\begin{align*}
M_{i,j} &= \exp(L)_{i,j} = L_{i,j} + \sum\limits_{m=2}^{n-1} \frac{1}{m!} (L^{m})_{i,j} \\ \Rightarrow L_{i,j} &= M_{i,j} - \sum\limits_{m=2}^{n-1} \frac{1}{m!} (L^{m})_{i,j} .
\end{align*}
The result now follows from the induction hypothesis and from our previous claim, as this argument can be used to both construct such a matrix $L$ and to prove that it is uniquely determined.
\end{proof}

\subsection{Properties of commuting matrices}

We will now present a useful decomposition of $\mathbb{C}^{n}$ induced by the commuting matrices $A_{1}, \ldots, A_{k} \in \mathbb{C}^{n \times n}$. Let $\sigma(A_{i})$ denote the spectrum of the matrix $A_{i}$. In what follows, let
\begin{align*}
\boldsymbol{\lambda} = (\lambda_{1}, \ldots, \lambda_{k}) \in \sigma(A_{1}) \times \cdots \times \sigma(A_{k}) .
\end{align*}
We remind the reader that $\ker(A_{i} - \lambda_{i})^{n}$ corresponds to the generalised eigenspace of $\lambda_{i}$ of $A_{i}$. Moreover, we define the following subspaces:
\begin{align*}
\mathcal{V}_{\boldsymbol{\lambda}} = \bigcap \limits_{i=1}^{k} \ker(A_{i} - \lambda_{i} I)^{n}.
\end{align*}
Also, let $\Sigma = \lbrace \boldsymbol{\lambda} \in \sigma(A_{1}) \times \cdots \times \sigma(A_{k}) : \mathcal{V}_{\boldsymbol{\lambda}} \neq \lbrace \boldsymbol{0} \rbrace \rbrace$.

\begin{theorem}
\label{subspace_decomposition}
For all $\boldsymbol{\lambda} = (\lambda_{1}, \ldots, \lambda_{k}) \in \Sigma$ and for all $i \in \lbrace 1, \ldots, k \rbrace$, the following properties hold:

\begin{enumerate}

\item $\mathcal{V}_{\boldsymbol{\lambda}}$ is invariant under $A_{i}$.

\item $\sigma(A_{i} \restriction_{\mathcal{V}_{\boldsymbol{\lambda}}}) = \lbrace \lambda_{i} \rbrace$.

\item $\mathbb{C}^{n} = \bigoplus \limits_{\boldsymbol{\lambda} \in \Sigma} \mathcal{V}_{\boldsymbol{\lambda}} .$

\end{enumerate}
\end{theorem}

\begin{proof}
We show, by induction on $k$, that the subspaces $\mathcal{V}_{\boldsymbol{\lambda}}$ satisfy the properties above.

When $k = 1$, the result follows from the existence of Jordan Canonical Forms. When $k > 1$, suppose that $\sigma(A_{k}) = \lbrace \mu_{1}, \ldots, \mu_{m} \rbrace$, and let $\mathcal{U}_{j} = \ker(A_{k} - \mu_{j} I)^{n}$, for $j \in \lbrace 1, \ldots, m \rbrace$. Again, it follows from the existence of Jordan Canonical Forms that
\begin{align*}
\mathbb{C}^{n} = \bigoplus \limits_{j = 1}^{m} \mathcal{U}_{m} .
\end{align*}
In what follows, $i \in \lbrace 1, \ldots, k-1 \rbrace$ and $j \in \lbrace 1, \ldots, m \rbrace$. Now, as $A_{k}$ and $A_{i}$ commute, so do $(A_{k}-\mu_{j} I)$ and $A_{i}$. Therefore, for all $\boldsymbol{v} \in \mathcal{U}_{j}$, $(A_{k} - \mu_{j} I)^{n} A_{i} \boldsymbol{v} = A_{i} (A-\mu_{j} I)^{n} \boldsymbol{v} = \boldsymbol{0}$, so $A_{i} \boldsymbol{v} \in \mathcal{U}_{j}$, that is, $\mathcal{U}_{j}$ is invariant under $A_{i}$. The result follows from applying the induction hypothesis to the commuting operators $A_{i} \restriction_{\mathcal{U}_{j}}$.
\end{proof}

We will also make use of the following well-known result on simultaneous triangularisation of commuting matrices. See, for example, \cite{CommutingMatrices}.

\begin{theorem}
\label{simultaneous-triangularisation}
Given $k$ commuting matrices $A_{1}, \ldots, A_{k} \in \overline{\mathbb{Q}}^{n \times n}$, there exists a matrix $P \in \mathit{GL}_{n}(\overline{\mathbb{Q}})$ such that $P^{-1}A_{i}P$ is upper triangular for all $i \in \lbrace 1, \ldots, k \rbrace$.
\end{theorem}

\subsection{Convex Polyhedra and Semi-Algebraic Sets}

A convex polyhedron is a subset of $\mathbb{R}^{n}$ of the form $\mathcal{P} = \lbrace \boldsymbol{x} \in \mathbb{R}^{n} : A \boldsymbol{x} \leq \boldsymbol{b} \rbrace$, where $A$ is a $d \times n$ matrix and $\boldsymbol{b} \in \mathbb{R}^{d}$. When all the entries of $A$ and coordinates of $\boldsymbol{b}$ are algebraic numbers, the convex polyhedron $\mathcal{P}$ is said to have an algebraic description.

A set $S \subseteq \mathbb{R}^{n}$ is said to be semi-algebraic if it is a Boolean combination of sets of the form $\lbrace \boldsymbol{x} \in \mathbb{R}^{n}: p(\boldsymbol{x}) \geq 0\rbrace$, where $p$ is a polynomial with integer coefficients. Equivalently, the semi-algebraic sets are those definable by the quantifier-free first-order formulas over the structure $(\mathbb{R}, <, +, \cdot, 0, 1)$.

It was shown by Alfred Tarski in \cite{Tarski} that the first-order theory of reals admits quantifier elimination. Therefore, the semi-algebraic sets are precisely the first-order definable sets.

\begin{theorem}[Tarski]
The first-order theory of reals is decidable.
\end{theorem}

See \cite{Renegar} and \cite{BPR06} for more efficient decision procedures for the first-order theory of reals.

\begin{definition}[Hilbert's Tenth Problem]
Given a polynomial $p \in \mathbb{Z}[x_{1}, \ldots, x_{k}]$, decide whether $p(\boldsymbol{x}) = 0$ admits a solution $\boldsymbol{x} \in \mathbb{N}^{k}$. Equivalently, given a semi-algebraic set $S \subseteq \mathbb{R}^{k}$, decide whether it intersects $\mathbb{Z}^{k}$.
\end{definition}

The following celebrated theorem, due to Yuri Matiyasevich, will be used in
our undecidability proof; see \cite{HTP} for a self-contained proof.

\begin{theorem}[Matiyasevich]
Hilbert's Tenth Problem is undecidable.
\end{theorem}

On the other hand, our proof of decidability of ALIP makes use of some techniques present in the proof of the following result, shown in \cite{KP}:

\begin{theorem}[Khachiyan and Porkolab]
It is decidable whether a given \emph{convex} semi-algebraic set $S \subseteq \mathbb{R}^{k}$ intersects $\mathbb{Z}^{k}$.
\end{theorem}

\subsection{Fourier-Motzkin Elimination}

Fourier-Motzkin elimination is a simple method for solving systems of
inequalities. Historically, it was the first algorithm used in solving
linear programming, before more efficient procedures such as the
simplex algorithm were discovered. The procedure consists in isolating one
variable at a time and matching all its lower and upper bounds. Note
that this method preserves the set of solutions on the remaining
variables, so a solution of the reduced system can always be extended
to a solution of the original one.

\begin{theorem}
\label{thm:fme}
  By using Fourier-Motzkin elimination, it is decidable whether a
  given convex polyhedron
  $\mathcal{P} = \lbrace \boldsymbol{x} \in \mathbb{R}^{n} : \pi
  A\boldsymbol{x} < \boldsymbol{b} \rbrace$,
  where the entries of $A$ are all real algebraic numbers and
  those of $\boldsymbol{b}$ are real linear forms in logarithms of
  algebraic numbers, is empty.  Moreover, if $\mathcal{P}$ is
  non-empty one can effectively find a rational vector
  $\boldsymbol{q} \in \mathcal{P}$.
\end{theorem}

\begin{proof}
When using Fourier-Motzkin elimination, isolate each term $\pi x_{i}$, instead of just isolating the variable $x_{i}$. Note that the coefficients of the terms $\pi x_{i}$ will always be algebraic, and the loose constants will always be linear forms in logarithms of algebraic numbers, which are closed under multiplication by algebraic numbers, and which can be effectively compared by using Baker's Theorem.
\end{proof}
\section{Example}

Let $\lambda_{1}, \lambda_{2} \in \mathbb{R} \cap \bar{\mathbb{Q}}$ such that $\lambda_{1} > \lambda_{2}$ and consider the following commuting matrices $A_{1}, A_{2} \in (\mathbb{R} \cap \bar{\mathbb{Q}})^{2 \times 2}$:

\begin{align*}
A_{i} = \begin{pmatrix} \lambda_{i} && 1 \\ 0 && \lambda_{i} \end{pmatrix}, i \in \lbrace 1, 2 \rbrace .
\end{align*}

One can easily see that
\begin{align*}
\exp(A_{i} t_{i}) &= \exp(\lambda_{i} t_{i} I) \exp(t_{i} (A_{i} - \lambda_{i} I)) \\
&= \exp(\lambda_{i} t_{i})
\exp \begin{pmatrix} 0 && t_{i} \\ 0 && 0 \end{pmatrix} \\
&= \exp (\lambda_{i} t_{i})
\begin{pmatrix} 1 && t_{i} \\ 0 && 1 \end{pmatrix}, i \in \lbrace 1, 2 \rbrace .
\end{align*}

Let $c_{1}, c_{2} \in \mathbb{R} \cap \bar{\mathbb{Q}}$ such that $c_{1}, c_{2} > 0$, and let
\begin{align*}
C = \begin{pmatrix} c_{1} && c_{2} \\ 0 && c_{1} \end{pmatrix} .
\end{align*}

We would like to determine whether there exists a solution $t_{1}, t_{2} \in \mathbb{R}$, $t_{1}, t_{2} \geq 0$ to
\begin{align*}
\exp(A_{1} t_{1}) \exp(A_{2} t_{2}) = C
\end{align*}

This amounts to solving the following system of equations:
\begin{align*}
&\begin{cases}
\exp(\lambda_{1} t_{1} + \lambda_{2} t_{2}) = c_{1} \\
(t_{1} + t_{2}) \exp(\lambda_{1} t_{1} + \lambda_{2} t_{2}) = c_{2}
\end{cases}
&\Leftrightarrow \\
&\begin{cases}
\exp(t_{1} (\lambda_{1} - \lambda_{2}) + \frac{c_{2}}{c_{1}} \lambda_{2}) = c_{1} \\
t_{2} = \frac{c_{2}}{c_{1}} - t_{1}
\end{cases}
&\Leftrightarrow \\
&\begin{cases}
t_{1} = \frac{\log(c_{1}) - \frac{c_{2}}{c_{1}} \lambda_{2}}{\lambda_{1} - \lambda_{2}} \\
t_{2} = \frac{\frac{c_{2}}{c_{1}} \lambda_{1} - \log(c_{1})}{\lambda_{1} - \lambda_{2}}
\end{cases}
\end{align*}

Then $t_{1}, t_{2} \geq 0$ holds if and only if
\begin{align*}
\lambda_{2} \leq \frac{c_{1}}{c_{2}} \log(c_{1}) \leq \lambda_{1} .
\end{align*}

Whether these inequalities hold amounts to comparing linear forms in logarithms of algebraic numbers.
\section{Decidability in the Commutative Case}

We start this section by reducing the Generalised MEP with commuting
matrices to LEP.  The intuition behind it is quite simple: perform a
change of basis so that the matrices $A_{1}, \ldots, A_{k}$, as well
as $C$, become block-diagonal matrices, with each block being upper
triangular; we can then separate the problem into several
sub-instances, corresponding to the diagonal blocks, and finally make
use of our uniqueness result concerning strictly upper triangular
logarithms of upper unitriangular matrices.

\begin{theorem}
  The Generalised MEP with commuting matrices reduces to LEP.
\end{theorem}

\begin{proof}
  Consider an instance of the generalised MEP, as given in Definition~\ref{def:MEP},
  with commuting $n\times n$ matrices $A_1,\ldots,A_k$ and target
  matrix $C$.

  We first show how to define a matrix $P$ such that each matrix
  $P^{-1}A_iP$ is block diagonal, $i=1,\ldots,k$, with each block
  being moreover upper triangular.

  By Theorem~\ref{subspace_decomposition} we can write $\mathbb{C}^n$
  as a direct sum of subspaces $\mathbb{C}^n = \oplus_{j=1}^b \mathcal{V}_j$
  such that for every subspace $\mathcal{V}_j$ and matrix $A_i$, $\mathcal{V}_j$ is an
  invariant subspace of $A_i$ on which $A_i$ has a single eigenvalue
  $\lambda_i^{(j)}$.
 
  Define a matrix $Q$ by picking an algebraic basis for each
  $\mathcal{V}_j$ and successively taking the vectors of each basis to
  be the columns of $Q$. Then, each matrix $Q^{-1} A_{i} Q$ is
  block-diagonal, where the $j$-th block is a matrix $B^{(j)}_i$ that
  represents $A_{i} \restriction{\mathcal{V}_j}$, $j=1,\ldots,b$.  

  Fixing $j\in\{1,\ldots,b\}$, note that the
  matrices $B_1^{(j)},\ldots,B_k^{(j)}$ all commute.  Thus we
  may apply Theorem \ref{simultaneous-triangularisation} to obtain an
  algebraic matrix $M_j$ such that each matrix $M_j^{-1} B^{(j)}_{i} M_j$
  is upper triangular, $i=1,\ldots,k$.  Thus we can write
  \[ M_j^{-1} B^{(j)}_{i} M_j = \lambda_i^{(j)}I + N_i^{(j)} \]
  for some strictly upper triangular matrix $ N_i^{(j)}$.

  We define $M$ to be the block-diagonal matrix with blocks $M_1,\ldots,M_b$.
  Letting $P=QM$, it is then the case
  that $P^{-1} A_{i} P$ is block-diagonal, with the $j$-th block being
  $\lambda_i^{(j)}I + N_i^{(j)}$ for $j=1,\ldots,b$.  Now
\begin{align}
\prod \limits_{i=1}^{k} \exp(A_{i} t_{i}) = C \Leftrightarrow \prod \limits_{i=1}^{k} \exp(P^{-1}A_{i}P t_{i}) = P^{-1}CP .
\label{eq:block}
\end{align}

If $P^{-1}CP$ is not block-diagonal, with each block being upper
triangular and with the same entries along the diagonal, then Equation
(\ref{eq:block}) has no solution and the problem instance must be
negative. Otherwise, denoting the blocks $P^{-1}CP$ by $D^{(j)}$ for
$j \in \lbrace 1, \ldots, b \rbrace$, our problem amounts to
simultaneously solving the system of matrix equations
\begin{align}
\prod\limits_{i=1}^{k} \exp\big(\big(\lambda_i^{(j)}I + N_i^{(j)}\big)t_{i}\big) = D^{(j)}, \quad j \in \lbrace 1, \ldots, b \rbrace
\label{eq:main1}
\end{align}
with one equation for each block.

For each fixed $j$, the matrices $N_{i}^{(j)}$ inherit commutativity from
the matrices $B^{(j)}_{i}$, so we have
\begin{align*}
\prod\limits_{i=1}^{k} \exp((\lambda_i^{(j)}I + N_i^{(j)})t_{i}) &=
   \exp\big(\sum_{i=1}^k (\lambda_i^{(j)}I  + 
 N_i^{(j)}) t_i \big)\\
&= \exp\big(\sum_{i=1}^k \lambda_i^{(j)} t_i\big) \cdot
   \exp\big(\sum_{i=1}^k N_i^{(j)} t_i \big) .
\end{align*}

Hence the system (\ref{eq:main1}) is equivalent to
\begin{align} 
\exp\big(\sum_{i=1}^k \lambda_i^{(j)} t_i\big) \cdot
   \exp\big(\sum_{i=1}^k N_i^{(j)} t_i \big)  = D^{(j)}
\label{eq:main2}
\end{align}
for $j=1,\ldots,b$.

By assumption, the diagonal entries of each matrix $D^{(j)}$ are
equal to a unique value, say $c^{(j)}$. 
Since the diagonal entries of
$\exp\left(\sum_{i=1}^kN^{(j)}t_i\right)$
are all $1$, the equation system (\ref{eq:main2}) is equivalent to:
\begin{align*}
\exp\big(\sum_{i=1}^k \lambda_i^{(j)} t_i\big)
= c^{(j)} \mbox{ and }\exp\big(\sum_{i=1}^k N_i^{(j)} t_i \big)
=\frac{1}{c^{(j)}} D^{(j)}
\end{align*}
for $j=1,\ldots,b$.

Applying Theorem \ref{logarithm_uniqueness}, the
above system can equivalently be written
\begin{align*}
\exp\big(\sum_{i=1}^k \lambda_i^{(j)} t_i\big)
= c^{(j)} \mbox{ and } \sum_{i=1}^k
N_i^{(j)} t_i =
S^{(j)}
\end{align*}
for  some effectively computable matrix
$S^{(j)}$ with algebraic entries, $j=1,\ldots,b$.

Except for the additional linear equations, this has the form of an
instance of LEP.
 However we can eliminate the linear equations by
performing a linear change of variables, i.e., by computing the
solution of the system in parametric form.  Thus we finally arrive at
an instance of LEP.
\end{proof}

In the following result, we essentially solve the system of equations \ref{single_eqn_form}, reducing it to the simpler problem that really lies at its heart.

\begin{theorem}
\label{reference-for-log}
LEP reduces to ALIP.
\end{theorem}

\begin{proof}
Consider an instance of LEP, comprising a system of equations
\begin{align}
 \exp\left(\sum_{\ell=1}^k \lambda_\ell^{(j)} t_\ell \right) = c_j \exp (d_j) 
\quad j=1,\ldots,b,
\label{eq:LEPinstance}
\end{align}
and polyhedron $\mathcal{P}\subseteq \mathbb{R}^{2k}$, as described in
Definition~\ref{def:LEP}.

Throughout this proof, let $\log$ denote a fixed logarithm branch that
is defined on all the numbers $c_j, \exp(d_j)$ appearing
above, and for which $\log(-1) = i \pi$. Note that if any $c_j=0$ for
some $j$ then (\ref{eq:LEPinstance}) has no solution. Otherwise, by
applying $\log$ to each equation in (\ref{eq:LEPinstance}),
we get:
\begin{align}
\sum_{\ell=1}^k \lambda_\ell^{(j)} t_\ell = d_j+\log(c_j) + 2i\pi n_j \quad j=1,\ldots,b,
\label{eq:logs}
\end{align}
where $n_j \in \mathbb{Z}$.

The system of equations (\ref{eq:logs}) can be written in matrix form as
\begin{align*}
A \boldsymbol{t} \in \boldsymbol{d}+\log(\boldsymbol{c}) + 
2i\pi \mathbb{Z}^b \, ,
\end{align*}
where $A$ is the $b\times k$ matrix with $A_{j,\ell} = \lambda_\ell^{(j)}$ and $\log$
is applied pointwise to vectors.
Now, defining the convex polyhedron $\mathcal{Q}\subseteq \mathbb{R}^{2b}$ by
\begin{align*}
\mathcal{Q} = \lbrace &(\Re(A\boldsymbol{y}), \Im(A\boldsymbol{y})) : 
\boldsymbol{y}\in\mathbb{C}^k, (\Re(\boldsymbol{y}), \Im(\boldsymbol{y})) \in \mathcal{P} \rbrace \, ,
\end{align*}
it suffices to decide whether the affine lattice
$\boldsymbol{d} + \log(\boldsymbol{c})+ 2i\pi  \mathbb{Z}^b$
intersects
$\lbrace \boldsymbol{x} \in \mathbb{C}^b : (\Re(\boldsymbol{x}),
\Im(\boldsymbol{x})) \in \mathcal{Q} \rbrace$.

Define $f:\mathbb{R}^b \rightarrow \mathbb{C}^b$ by
$f(\boldsymbol{v})= \boldsymbol{d} + \log(\boldsymbol{c}) +
2i\pi \boldsymbol{v}$,
and define a convex polyhedron $\mathcal{T}\subseteq \mathbb{R}^b$ by
\[\mathcal{T}=\{ \boldsymbol{v}\in\mathbb{R}^b : (\Re(f(\boldsymbol{v})),
\Im (f(\boldsymbol{v}))) \in \mathcal{Q} \} \, . \]

The problem then amounts to deciding whether the convex polyhedron
$\mathcal{T}$ intersects contains an integer point. Crucially, the
description of the convex polyhedron $\mathcal{T}$ is of the form
$\pi B\boldsymbol{x} \leq \boldsymbol{b}$, for some matrix $B$ and
vector $\boldsymbol{b}$ such that the entries of $B$ are real
algebraic and the components of $\boldsymbol{b}$ are real linear forms
in logarithms of algebraic numbers.  But this is the form of an
instance of ALIP.
\end{proof}

We are left with the task of showing that ALIP is decidable. The argument essentially consists of reducing to a lower-dimensional instance whenever possible, and eventually either using the fact that the polyhedron is bounded to test whether it intersects the integer lattice or using Kronecker's theorem to show that, by a density argument, it must intersect the integer lattice.

\begin{theorem}
ALIP is decidable.
\end{theorem}

\begin{proof}
We are given a convex polyhedron $\mathcal{P}=\lbrace \boldsymbol{x} \in \mathbb{R}^{d} : \pi A \boldsymbol{x} \leq \boldsymbol{b} \rbrace$, where the coordinates $\boldsymbol{b}$ are linear forms in logarithms of algebraic numbers, and need to decide whether this polyhedron intersects $\mathbb{Z}^{d}$. Throughout this proof, $\log$ denotes the logarithm branch picked at the beginning of the proof of Theorem \ref{reference-for-log}. We start by eliminating linear dependencies between the logarithms appearing therein, using Masser's Theorem. For example, suppose that
\begin{align*}
b_{i} = r_{0} + r_{1} \log(s_{1}) + \cdots + r_{k} \log(s_{k}), r_{0}, r_{1}, s_{1}, \ldots, r_{k}, s_{k} \in \overline{\mathbb{Q}}.
\end{align*}
Due to Baker's theorem, there exists a non-trivial linear relation with algebraic coefficients amongs $\log(-1), \log(s_{1}), \ldots, \log(s_{k})$ if and only if there is one with integer coefficients. But such relations can be computed, since
\begin{align*}
&n_{0} \log(-1) + n_{1} \log(s_{1}) + \cdots + n_{k} \log(s_{k}) = 0 \Leftrightarrow \\
&(-1)^{n_{0}} s_{1}^{n_{1}} \cdots s_{k}^{n_{k}} = 1
\end{align*}
and since the group of multiplicative relations $L(-1, s_{1}, \ldots, s_{k})$ can be effectively computed. Whenever it contains a non-zero vector, we use it to eliminate an unnecessary $\log(s_{i})$ term, although never eliminating $\log(-1)$. When this process is over, we can see whether each term $b_{i}/\pi$ is algebraic or transcendental: it is algebraic if $b_{i} = \alpha \log(-1), \alpha \in \overline{\mathbb{Q}}$, and transcendental otherwise.

Now, when $\boldsymbol{x} \in \mathbb{Z}^{d}$, $A \boldsymbol{x}$ is a vector with algebraic coefficients, so whenever $b_{i} / \pi$ is transcendental we may alter $\mathcal{P}$ by replacing $\leq$ by $<$ in the $i$-th inequality, preserving its intersection with $\mathbb{Z}^{d}$. On the other hand, whenever $b_{i} / \pi$ is algebraic, we split our problem into two: in the first one, $\mathcal{P}$ is altered to force equality on the $i$-th constraint (that is, replacing $\leq$ by $=$), and in the second we force strict inequality (that is, replacing $\leq$ by $<$). We do this for all $i$, so that no $\leq$ is left in any problem instance, leaving us with finitely many polyhedra, each defined by equations of the form
\begin{align*}
K \boldsymbol{x} &= \boldsymbol{k} \quad &(\boldsymbol{k} \in \overline{\mathbb{Q}}^{d_{1}}) \\
M \boldsymbol{x} &< \boldsymbol{m} \quad &(\boldsymbol{m} \in \overline{\mathbb{Q}}^{d_{2}}) \\
F \boldsymbol{x} &< \boldsymbol{f} \quad &(\boldsymbol{f} \in \mathbb{R} \setminus \overline{\mathbb{Q}}^{d_{3}})
\end{align*}
where $K,M,F$ are matrices with algebraic entries. Before proceeding, we eliminate all such empty polyhedra; note that emptiness can be decided via Fourier-Motzkin elimination, as shown in Theorem \ref{thm:fme}.

The idea of the next step is to reduce the dimension of all the problem instances at hand until we are left with a number of new instances with full-dimensional open convex polyhedra, of the same form as the original one, apart from the fact that all inequalities in their definitions will be strict. To do that, we use the equations $K \boldsymbol{x} = \boldsymbol{k}$ to eliminate variables: note that, whenever there is an integer solution,
\begin{align*}
K \boldsymbol{x} = \boldsymbol{k}, \boldsymbol{x} \in \mathbb{Z}^{d} \Leftrightarrow \boldsymbol{x} = \boldsymbol{x}_{0} + M \boldsymbol{z},
\end{align*}
where $M$ is a matrix with integer entries, $\boldsymbol{x}_{0}$ is an integer vector and $\boldsymbol{z}$ ranges over integer vectors over a smaller dimension space, wherein we also define the polyhedron
\begin{align*}
\mathcal{Q} = \lbrace \boldsymbol{y} : \boldsymbol{x}_{0} + M \boldsymbol{y} \in \mathcal{P} \rbrace .
\end{align*}

Having now eliminated all equality constraints, we are left with a finite set of polyhedra of the form $\mathcal{P} = \lbrace \boldsymbol{x} \in \mathbb{R}^{d}: \pi A \boldsymbol{x} < \boldsymbol{b} \rbrace$ that are either empty or full-dimensional and open, and wish to decide whether they intersect the integer lattice of the corresponding space (different instances may lie in spaces of different dimensions, of course). Note that, when $\mathcal{P}$ is non-empty, we can use Fourier-Motzkin elimination to find a vector $\boldsymbol{q} \in \mathbb{Q}^{d}$ in its interior, and $\varepsilon > 0$ such that the $l_{1}$ closed ball of radius $\varepsilon$ and centre $\boldsymbol{q}$, which we call $\mathcal{B}$, is contained in $\mathcal{P}$.

The next step is to consider the Minkowski-Weyl decomposition of $\mathcal{P}$, namely $\mathcal{P} = \mathcal{H} + \mathcal{C}$, where $\mathcal{H}$ is the convex hull of finitely many points of $\mathcal{P}$ (which we need not compute) and $\mathcal{C} = \lbrace \boldsymbol{x} \in \mathbb{R}^{d}: A \boldsymbol{x} \leq \boldsymbol{0} \rbrace$ is a cone with an algebraic description. Note that $\mathcal{P}$ is bounded if and only if $\mathcal{C} = \lbrace \boldsymbol{0} \rbrace$, in which case the problem at hand is simple: consider the polyhedron $\mathcal{Q}$ with an algebraic description obtained by rounding up each coordinate of $\boldsymbol{b} / \pi$, which has the same conic part as $\mathcal{P}$ and which contains $\mathcal{P}$, and therefore is bounded; finally, compute a bound on $\mathcal{Q}$ (such a bound can be defined in the first-order theory of reals), which is also a bound on $\mathcal{P}$, and test the integer points within that bound for membership in $\mathcal{P}$. Otherwise,
\begin{align*}
\mathcal{C} = \lbrace \lambda_{1} \boldsymbol{c}_{1} + \cdots + \lambda_{k} \boldsymbol{c}_{k}: \lambda_{1}, \ldots, \lambda_{k} \geq 0 \rbrace,
\end{align*}
where $\boldsymbol{c}_{1}, \ldots, \boldsymbol{c}_{k} \in \overline{\mathbb{Q}}^{d}$ are the extremal rays of $\mathcal{C}$. Note that $\boldsymbol{q} + \mathcal{C} \subseteq \mathcal{P}$ and that $\mathcal{B} + \mathcal{C} \subseteq \mathcal{P}$.

Now we consider a variation of an argument which appears in \cite{KP}. Consider the computable set
\begin{align*}
\mathcal{L} = \mathcal{C}^{\perp} \cap \mathbb{Z}^{d} = \bigcap\limits_{i=1}^{k} A(\boldsymbol{c}_{i}) ,
\end{align*}
where $A(\boldsymbol{v})$ denotes the group of additive relations of $\boldsymbol{v}$.

If $\mathcal{L} = \lbrace \boldsymbol{0} \rbrace$ then due to Kronecker's theorem on simultaneous Diophantine approximation it must be the case that there exists a vector $(n_{1}, \ldots, n_{k}) \in \mathbb{N}^{k}$ such that
\begin{align*}
\operatorname{dist} \left(\boldsymbol{q} + \sum\limits_{i=1}^{k} n_{i} \boldsymbol{c}_{i}, \mathbb{Z}^{d} \right) \leq \varepsilon ,
\end{align*}
and we know that $\mathcal{P} \cap \mathbb{Z}^{d} \neq \emptyset$ from the fact that the $l_{1}$ closed ball $\mathcal{B}$ of radius $\varepsilon$ and centre $\boldsymbol{q}$ is contained in $\mathcal{P}$.

On the other hand, if
$\mathcal{L} \neq \lbrace \boldsymbol{0} \rbrace$, let
$\boldsymbol{z} \in \mathcal{L} \setminus \lbrace \boldsymbol{0}
\rbrace$.
Since $\mathcal{H}$ is a bounded subset of $\mathbb{R}^n$, the set
\[ \{ \boldsymbol{z}^T\boldsymbol{x} : \boldsymbol{x}\in \mathcal{P} \} =
   \{ \boldsymbol{z}^T\boldsymbol{x} : \boldsymbol{x}\in \mathcal{H} \} \]
is a bounded subset of $\mathbb{R}$.
Therefore there exist
$a,b \in \mathbb{Z}$ such that
\begin{align*}
\forall \boldsymbol{x} \in \mathcal{P}, a \leq \boldsymbol{z}^{T} \boldsymbol{x} \leq b ,
\end{align*}
so we can reduce our problem to $b-a+1$ smaller-dimensional instances
by finding the integer points of
$\lbrace \boldsymbol{x} \in \mathcal{P} :
\boldsymbol{z}^{T} \boldsymbol{x} = i \rbrace$,
for $i \in \lbrace a, \ldots, b \rbrace$. Note that we have seen
earlier in the proof how to reduce the dimension of the ambient space
when the polyhedron $\mathcal{P}$ is contained in an affine
hyperplane.
\end{proof}

\section{Undecidability of the Non-Commutative Case}

In this section we show that the Matrix-Exponential Problem is
undecidable in the case of non-commuting matrices.  We show
undecidability for the most fundamental variant of the problem, as
given in Definition~\ref{def:MEP}, in which the matrices have real
entries and the variables $t_i$ range over the non-negative reals.
Recall that this problem is decidable in the commutative case by the
results of the previous section.

\subsection{Matrix-Exponential Problem with Constraints}

The proof of undecidability in the non-commutative case is by
reduction from Hilbert's Tenth Problem.  The reduction proceeds via
several intermediate problems.  These problems are obtained by
augmenting MEP with various classes of arithmetic constraints on the
real variables that appear in the statement of the problem.

\begin{definition}
  We consider the following three classes of arithmetic constraints
  over real variables $t_1,t_2,\ldots$:
\begin{itemize}
\item $\mathcal{E}_{\pi\mathbb{Z}}$ comprises constraints of the form
  $t_i\in\alpha+\beta\pi\mathbb{Z}$, where $\alpha$ and $\beta\neq 0$
  are real-valued constants such that $\cos(2\alpha\beta^{-1})$,
  $\beta$ are both algebraic numbers.
\item $\mathcal{E}_{+}$ comprises linear equations of the form
  $\alpha_1 t_1 + \ldots + \alpha_n t_n = \alpha_0 $, for
  $\alpha_0,\ldots,\alpha_n$ real algebraic constants.
\item $\mathcal{E}_{\times}$ comprises equations of the form
  $t_\ell=t_it_j$.
\end{itemize}
\end{definition}

A class of constraints $\mathcal{E} \subseteq \mathcal{E}_{\pi\mathbb{Z}} \cup
\mathcal{E}_{+}\cup \mathcal{E}_{\times}$ 
induces a generalisation of the MEP problem as follows:
\begin{definition}[MEP with Constraints]
  Given a class of constraints
  $\mathcal{E} \subseteq \mathcal{E}_{\pi\mathbb{Z}} \cup
  \mathcal{E}_{+}\cup \mathcal{E}_{\times}$,
  the problem MEP$(\mathcal{E})$ is as follows.  An instance consists
  of real algebraic matrices $A_1,\ldots,A_k,C$ and a finite set of
  constraints $E\subseteq\mathcal{E}$ on real variables
  $t_1,\ldots,t_k$.  The question is whether there exist non-negative
  real values for $t_1,\ldots,t_k$ such that
  $\prod_{i=1}^ke^{A_it_i}=C$ and the constraints $E$ are all
  satisfied.
\label{def:contraintMEP}
\end{definition}

Note that in the above definition of MEP$(\mathcal{E})$ the set of
constraints $E$ only mentions real variables $t_1,\ldots,t_k$
appearing in the matrix equation $\prod_{i=1}^ke^{A_it_i}=C$.
However, without loss of generality, we can allow constraints to
mention fresh variables $t_i$, for $i>k$, since we can always define a
corresponding matrix $A_i=0$ for such variables for then
$e^{A_it_i}=I$ has no effect on the matrix product.  In other words,
we effectively have constraints in $\mathcal{E}$ with existentially
quantified variables.  In particular, we have the following
useful observations:

\begin{itemize}
\item[\textbullet] We can express inequality constraints of the form
  $t_i\neq \alpha$ in  $\mathcal{E}_{+}\cup \mathcal{E}_{\times}$ by
using fresh variables
  $t_j,t_\ell$.  Indeed $t_i \neq \alpha$ is satisfied whenever there
  exist values of $t_j$ and $t_{\ell}$ such that $t_i=t_j+\alpha$ and
  $t_jt_\ell=1$.

\item[\textbullet] By using fresh variables,
  $\mathcal{E}_{+}\cup \mathcal{E}_{\times}$ can express polynomial
  constraints of the form $P(t_1,\ldots,t_n)=t$ for $P$ a polynomial
  with integer coefficients.
\end{itemize}

We illustrate the above two observations in an example.
\begin{example}
  Consider the problem, given matrices $A_1,A_2$ and $C$, to decide
  whether there exist $t_1,t_2 \geq 0$ such that
  \[ e^{A_1t_1}e^{A_2t_2}=C \,\mbox{ and }\, t_1^2-1=t_2, t_2\neq 0 \, .\]
  This is equivalent to the following instance of
  MEP$(\mathcal{E}_{+}\cup\mathcal{E}_{\times})$: decide whether there exist $t_1,\ldots,t_5\geq 0$ such that
\[ \prod_{i=1}^5 e^{A_it_i}=C \,\mbox{ and }\, t_1t_1=t_3, t_3-1=t_2, t_2t_4=t_5, t_5=1\]
where $A_1,A_2$ and $C$ are as above and $A_3=A_4=A_5=0$.
% where
% \begin{equation*}
% A_1=\begin{bmatrix}1&2\\0&1\end{bmatrix}, A_2=\begin{bmatrix}0&3\\2&1\end{bmatrix},
% A_3=\begin{bmatrix}-1&\sqrt{2}\\7&\tfrac{3}{42}\end{bmatrix},C=\begin{bmatrix}4&3\\2&1\end{bmatrix}.
% \end{equation*}
\end{example}

% The following observations will prove useful in the following:
% \begin{itemize}
% \item[\textbullet] Since $\mathcal{E}_{+}$ contains equations of the
%   form $t_i=t_j$ and $t_i=\alpha$ and so we can use constraints to fix
%   the value of a variable to be a real algebraic constant and we can
%   ensure that several variables have the same value.

% \end{itemize}

We will make heavy use of the following proposition to combine several
instances of the constrained MEP into a single instance by combining
matrices block-wise.
\begin{proposition}\label{prop:remark}
  Given real algebraic matrices $A_1,\ldots,A_k,C$ and
  $A_1',\ldots,A_k',C'$, there exist real algebraic matrices
  $A_1'',\ldots,A_k''$, $C''$ such that for all $t_1,\ldots,t_k$:
\[\prod_{i=1}^ke^{A_i''t_i}=C''\qquad\Leftrightarrow\qquad\prod_{i=1}^ke^{A_it_i}=C\wedge\prod_{i=1}^ke^{A_i't_i}=C'.\]
\label{prop:combine}
\end{proposition}

\begin{proof}
Define for any $i\in\{1,\ldots,k\}$:
\[A_i''=\begin{bmatrix}A_i&0\\0&A_i'\end{bmatrix},\qquad
C''=\begin{bmatrix}C&0\\0&C'\end{bmatrix}.\]
The result follows because the matrix exponential can be computed
block-wise (as is clear from its power series definition):
\[\prod_{i=1}^ke^{A_i''t_i}=\prod_{i=1}^k\begin{bmatrix}e^{A_it_i}&0\\0&e^{A_i't_i}\end{bmatrix}=
\begin{bmatrix}\prod_{i=1}^ke^{A_it_i}&0\\0&\prod_{i=1}^ke^{A_i't_i}\end{bmatrix}.\]
\end{proof}

We remark that in the statement of Proposition~\ref{prop:remark} the
two matrix equations that are combined are over the same set of
variables.  However, we can clearly combine any two matrix equations
for which the common variables appear in the same order in the
respective products.

The core of the reduction is to show that the constraints in
$\mathcal{E}_{\pi\mathbb{Z}},\mathcal{E}_{+}$ and
$\mathcal{E}_{\times}$ do not make the MEP problem harder: one can
always encode them using the matrix product equation.

\begin{proposition}
  MEP$(\mathcal{E}_{\pi\mathbb{Z}} \cup \mathcal{E}_{+} \cup
  \mathcal{E}_{\times})$ reduces to MEP$(\mathcal{E}_{+} \cup
  \mathcal{E}_{\times})$.
\label{lem:pi}
\end{proposition}
\begin{proof}
  Let $A_1,\ldots,A_k,C$ be real algebraic matrices and
  $E\subseteq \mathcal{E}_{\pi\mathbb{Z}} \cup \mathcal{E}_{+} \cup
  \mathcal{E}_{\times}$
  a finite set of constraints on $t_1,\ldots,t_k$.  Since $E$ is
  finite it suffices to show how to eliminate from $E$ each constraint
  in $\mathcal{E}_{\pi \mathbb{Z}}$.

  Let $t_j\in\alpha+\beta\pi\mathbb{Z}$ be a constraint in $E$.  By definition 
of  $\mathcal{E}_{\pi\mathbb{Z}}$ 
we have
  that $\cos(2\alpha\beta^{-1}),\sin(2\alpha\beta^{-1})$ and
  $\beta\neq0$ are real algebraic.  Now define the following extra
  matrices:
\[A'_j=\begin{bmatrix}0&2\beta^{-1}\\-2\beta^{-1}&0\end{bmatrix},
 C'=\begin{bmatrix}\cos(2\alpha\beta^{-1})&
\sin(2\alpha\beta^{-1})\\-\sin(2\alpha\beta^{-1})&\cos(2\alpha\beta^{-1})\end{bmatrix}.\]
Our assumptions ensure that $A_j'$ and $C'$ are both real algebraic.

We now have the following chain of equivalences:
\begin{align*}
e^{A'_j t_j}=C'
&\Leftrightarrow\begin{bmatrix}\cos(2t_j\beta^{-1})&
\sin(2t_j\beta^{-1})\\-\sin(2t_j\beta^{-1})&\cos(2t_j\beta^{-1})\end{bmatrix}
=C'\\
&\Leftrightarrow\cos(2t_j\beta^{-1})=\cos(2\alpha\beta^{-1})\\
&\qquad\wedge\sin(2t_j\beta^{-1})=\sin(2\alpha\beta^{-1})\\
&\Leftrightarrow2\beta^{-1}t_j=2\alpha\beta^{-1}\mod 2\pi\\
&\Leftrightarrow t_j\in\alpha+\beta\pi\mathbb{Z}.
\end{align*}
Thus the additional matrix equation $e^{A_j't_j}=C'$ is equivalent to
the constraint $t_j\in\alpha+\beta\pi\mathbb{Z}$.  Applying
Proposition~\ref{prop:combine} we can thus eliminate this constraint.
\end{proof}

\begin{proposition}
  MEP$(\mathcal{E}_{+}\cup\mathcal{E}_{\times})$ reduces to
  MEP$(\mathcal{E}_{+})$.
\label{lem:times}
\end{proposition}

\begin{proof}
  Let $A_1,\ldots,A_k,C$ be real algebraic matrices and
  $E \subseteq\mathcal{E}_{+}\cup\mathcal{E}_{\times}$ a finite set of
  constraints on variables $t_1,\ldots,t_k$.  We proceed as above,
  showing how to remove from $E$ each constraint from
  $\mathcal{E}_{\times}$.  In so doing we potentially increase the
  number of matrices and add new constraints from $\mathcal{E}_{+}$.

  Let $t_l=t_i t_j$ be an equation in $E$.  To eliminate this equation
the first step is to introduce fresh  
variables $x,x',y,y',z$ and add the constraints
\[ t_i=x,\, t_j=y,\, t_\ell = z,\]
which are all in $\mathcal{E}_{+}$.  We now add a new matrix equation
over the fresh variables $x,x',y,y',z$ that is equivalent to the
constraint $xy=z$.  Since this matrix equation involves a new set of
variables we are free to the set the order of the matrix products,
which is crucial to express the desired constraint.

The key gadget is the following matrix product equation,
  which holds for  any $x,x',y,y',z\geqslant0$:
\begin{align*}
&\begin{bmatrix}1&0&-z\\0&1&0\\0&0&1\end{bmatrix}\begin{bmatrix}1&0&0\\0&1&-y'\\0&0&1\end{bmatrix}
\begin{bmatrix}1&x&0\\0&1&0\\0&0&1\end{bmatrix}\\
&\qquad\times\begin{bmatrix}1&0&0\\0&1&y\\0&0&1\end{bmatrix}
\begin{bmatrix}1&-x'&0\\0&1&0\\0&0&1\end{bmatrix}=
\begin{bmatrix}1&x-x'&z-xy\\0&1&y-y'\\0&0&1\end{bmatrix}.
\end{align*}

Notice that each of the matrices on the left-hand side of the above
equation has a single non-zero off-diagonal entry.  Crucially each
matrix of this form can be expressed as an exponential.  Indeed we can
write the above equation as a matrix-exponential product
\[ e^{B_1z} e^{B_2y'} e^{B_3x} e^{B_4y} e^{B_5x'} = \begin{bmatrix}1&x-x'&z-xy\\0&1&y-y'\\0&0&1\end{bmatrix} \]
for matrices
\[\begin{array}{r@{}cp{.1cm}r@{}c}
B_1=&\begin{bmatrix}0&0&-1\\0&0&0\\0&0&0\end{bmatrix} &&
B_2=&\begin{bmatrix}0&0&0\\0&0&-1\\0&0&0\end{bmatrix} \\
B_3=&\begin{bmatrix}0&1&0\\0&0&0\\0&0&0\end{bmatrix} &&
B_4=&\begin{bmatrix}0&0&0\\0&0&1\\0&0&0\end{bmatrix} \\
B_5=&\begin{bmatrix}0&-1&0\\0&0&0\\0&0&0\end{bmatrix} &&
\end{array}\]
Thus the constraint $xy=z$ can be expressed as 
\begin{gather}
e^{B_1z} e^{B_2y'} e^{B_3x} e^{B_4y} e^{B_5z'} = I \, .
\label{eq:fresh}
\end{gather}

Again, we can apply Proposition~\ref{prop:combine} to combine the
equation (\ref{eq:fresh}) with the matrix equation from the original
problem instance and thus encode the constraint $x=yz$.

\end{proof}

\begin{proposition}
  MEP$(\mathcal{E}_{+})$ reduces to MEP.
\label{lem:plus}
\end{proposition}

\begin{proof}
  Let $A_1,\ldots,A_k,C$ be real algebraic matrices and
  $E \subseteq\mathcal{E}_{+}$ a set of constraints.
  We proceed as above, showing how to eliminate each constraint from
  $E$ that lies in $\mathcal{E}_{+}$, while preserving the set of
  solutions of the instance.

Let $\beta+\sum_{i=1}^k\alpha_it_i=0$ be an equation in $E$.
Recall that $\beta,\alpha_1,\ldots,\alpha_k$ are real algebraic.
Define the extra matrices $A_1',\ldots,A_k'$ and $C'$ as follows:
\[A_i'=\begin{bmatrix}0&\alpha_i\\0&0\end{bmatrix},
\qquad C'=\begin{bmatrix}1&-\beta\\0&1\end{bmatrix}.\]
Our assumptions ensure that $A_1',\ldots,A_k'$ and $C'$ are all real algebraic.
Furthermore, the following extra product equation becomes:
\begin{align*}
\prod_{i=1}^ke^{A_i't_i}=C
&\Leftrightarrow\prod_{i=1}^k\begin{bmatrix}1&\alpha_it_i\\0&1\end{bmatrix}=\begin{bmatrix}1&-\beta\\0&1\end{bmatrix}\\
&\Leftrightarrow\sum_{i=1}^k\alpha_it_i=-\beta \, .
\end{align*}
\end{proof}

Combining Propositions~\ref{lem:pi}, \ref{lem:times}, and \ref{lem:plus} we
have:
\begin{proposition}
  MEP$(\mathcal{E}_{\pi\mathbb{Z}}\cup\mathcal{E}_{+}\cup\mathcal{E}_{\times})$
  reduces to MEP.
\end{proposition}

\subsection{Reduction from Hilbert's Tenth Problem}

\begin{theorem}
  MEP is undecidable in the non-commutative case.
\end{theorem}

\begin{proof}
  We have seen in the previous section that the problem
  MEP$(\mathcal{E}_{\pi\mathbb{Z}}\cup\mathcal{E}_{+}\cup\mathcal{E}_{\times})$
  reduces to MEP without constraints.  Thus it suffices to reduce
  Hilbert's Tenth Problem to
  MEP$(\mathcal{E}_{\pi\mathbb{Z}}\cup\mathcal{E}_{+}\cup\mathcal{E}_{\times})$.
  In fact the matrix equation will not play a role in the
  target of this reduction, only the additional constraints.

  Let $P$ be a polynomial of total degree $d$ in $k$ variables with
  integer coefficients. From $P$ we build a homogeneous polynomial $Q$, by
  adding a new variable $\lambda$:
  \[Q(\mathbf{x},\lambda)=\lambda^dP\left(\frac{x_1}{\lambda},\ldots,
    \frac{x_k}{\lambda}\right).\]
  Note that $Q$ still has integer coefficients. Furthermore, we have the relationship
\[Q(\mathbf{x},1)=P(\mathbf{x}).\]

As we have seen previously, it is easy to encode $Q$ with constraints,
in the sense that we can compute a finite set of constraints
$E_Q \subseteq \mathcal{E}_{+}\cup\mathcal{E}_{\times}$ mentioning variables 
$t_0,\ldots,t_m,\lambda$
such that $E$ is
satisfied if and only if $t_0=Q(t_1,\ldots,t_k,\lambda)$. Note that
$E_Q$ may need to mention variables other
than $t_1,\ldots,t_k$ to do that.  Another
finite set of equations
$E_{\pi\mathbb{Z}}\subseteq\mathcal{E}_{\pi\mathbb{Z}}$ is used to
encode that $t_1,\ldots,t_k,\lambda\in\pi\mathbb{Z}$. Finally,
$E_{=}\subseteq\mathcal{E}_{+}\cup\mathcal{E}_{\times}$ is used to
encode $t_{0}=0$ and $1\leqslant \lambda\leqslant4$.  The
latter is done by adding the polynomial equations $\lambda=1+\alpha^2$ and
$\lambda=4-\beta^2$ for some $\alpha$ and $\beta$. Finally we have the
following chain of equivalences:
\begin{align*}
&\exists t_0,\ldots,\lambda\geqslant0\text{ s.t. }E_Q\cup E_{\pi\mathbb{Z}}\cup E_{=}\text{ is satisfied }\\
&\qquad\Leftrightarrow\exists t_1,\ldots,\lambda\geqslant0\text{ s.t. } 0=Q(t_1,\ldots,t_k,\lambda)\\
    &\qquad\qquad\wedge t_1,\ldots,t_k,\lambda\in\pi\mathbb{Z}\wedge 1\leqslant\lambda\leqslant4\\
&\qquad\Leftrightarrow\exists n_1,\ldots,n_k\in\mathbb{N}\text{ s.t. } 0=Q(\pi n_1,\ldots,\pi n_k,\pi)\\
&\qquad\Leftrightarrow\exists n_1,\ldots,n_k\in\mathbb{N}\text{ s.t. } 0=\pi^d Q(n_1,\ldots,n_k,1)\\
&\qquad\Leftrightarrow\exists n_1,\ldots,n_k\in\mathbb{N}\text{ s.t. } 0=P(n_1,\ldots,n_k).\\
\end{align*}
\end{proof}
\section{Conclusion}

We have shown that the Matrix-Exponential Problem is undecidable in
general, but decidable when the matrices $A_{1}, \ldots, A_{k}$ commute.
This is analogous to what was known for the discrete version
of this problem, in which the matrix exponentials $e^{At}$ are
replaced by matrix powers $A^n$.

A natural variant of this problem is the following:
\begin{definition}[Matrix-Exponential Semigroup Problem]
  Given square matrices $A_{1}, \ldots, A_{k}$ and $C$, all of the
  same dimension and all with real algebraic entries, is $C$ a member
  of the matrix semigroup generated by
\begin{align*}
\lbrace \exp(A_{i} t_{i}) : t_{i} \geq 0 , i=1,\ldots,k \rbrace ?
\end{align*}
\end{definition}
When the matrices $A_1,\ldots,A_k$ all commute, the above problem is
equivalent to the Matrix-Exponential Problem, and therefore decidable. In the non-commutative case, the following result holds:
\begin{theorem}
The Matrix-Exponential Semigroup Problem is undecidable.
\end{theorem}
A proof will appear in a future journal version of this paper. This can be done by reduction from the Matrix-Exponential Problem, using a set of gadgets to force a desired order in the multiplication of the matrix exponentials.

It would also be interesting to look at possibly decidable
restrictions of the MEP/MESP, for example the case where $k=2$ with a
non-commuting pair of matrices, which was shown to be decidable for
the discrete analogue of this problem in \cite{MEHTP}. Bounding the dimension of the ambient vector space could also yield decidability, which has been partly accomplished in the discrete case in \cite{CK05}. Finally, upper bounding the complexity of our decision procedure for the commutative case would also be a worthwhile task.

\acks

The author Jo\~{a}o Sousa-Pinto would like to thank Andrew Kaan Balin for a productive discussion during the early stages of this work.

\bibliographystyle{abbrvnat}
\softraggedright
\bibliography{refs}

\end{document}